\documentclass[aps,pra, reprint,superscriptaddress]{revtex4-1}
\usepackage{graphicx}
\usepackage{amsthm}
\usepackage{amsfonts}
\usepackage{amssymb}
\usepackage{array}
\usepackage{amsmath}
\usepackage{verbatim} 
\usepackage{hyperref}
\usepackage{color}
\usepackage{bbold}
\usepackage{epstopdf}
\usepackage{mathtools}
\usepackage{ulem}

\newtheorem{proposition}{Proposition}
\newtheorem{definition}{Definition}
\newtheorem{theorem}{Theorem}

\newcommand{\ket}[1]{\left\vert#1\right\rangle}
\newcommand{\bra}[1]{\left\langle#1\right\vert}

\def\bra#1{\langle #1|}
\def\ket#1{\left|#1 \right>}

\def\Tr{\mbox{Tr}}

\begin{document}
\title{Coherence, Asymmetry, and Quantum Macroscopicity}
\author{Hyukjoon Kwon}
\affiliation{Center for Macroscopic Quantum Control, Department of Physics and Astronomy, Seoul National University, Seoul, 151-742, Korea}
\author{Chae-Yeun Park}
\affiliation{Center for Macroscopic Quantum Control, Department of Physics and Astronomy, Seoul National University, Seoul, 151-742, Korea}
\affiliation{Asia Pacific Center for Theoretical Physics, Pohang, 37673, Korea}
\author{Kok Chuan Tan}
\affiliation{Center for Macroscopic Quantum Control, Department of Physics and Astronomy, Seoul National University, Seoul, 151-742, Korea}
\author{Daekun Ahn}
\affiliation{Center for Macroscopic Quantum Control, Department of Physics and Astronomy, Seoul National University, Seoul, 151-742, Korea}
\author{Hyunseok Jeong}
%\email{h.jeong37@gmail.com}
\affiliation{Center for Macroscopic Quantum Control, Department of Physics and Astronomy, Seoul National University, Seoul, 151-742, Korea}
\date{\today}

\begin{abstract}
We investigate a measure of quantum coherence and its extension to quantify quantum macroscopicity. 
The coherence measure can also quantify the asymmetry of a quantum state with respect to a given group transformation.
We then show that a weighted sum of asymmetry in each mode can be applied as a measure of macroscopic coherence.
To exclude the effects of microscopic superpositions, we suggest a method to introduce a cutoff to the weighted sum that will specify the macroscopic portion of the coherence. This cutoff may be interpreted as the fuzziness for a given measurement outcome. Based on the suggested measures, we investigate the quantum macroscopicity for particular concrete examples in $N$-partite spin systems.
\end{abstract}
\maketitle

\section{Introduction}  
Quantum theory has been very  successful in describing the properties of microscopic systems based on the superposition principle. 
Quantum coherence, which has its origin in the superposition principle, provides useful operational applications in quantum key distribution \cite{BB84, Ekert91,QCRMP}, computation \cite{Feynman82, Shor97, Deutsch92, Ladd10, Tan17}, and communication \cite{Bennett92, Bennett93, Duan10}.
Recently, there have been attempts to interpret quantum properties as a resource.
Resource theories based on the notion of free operations may provide useful tools to quantify nonclassical phenomena, covering a diverse range of topics that include entanglement \cite{Plenio07, Horodecki09}, asymmetry \cite{Gour08, Marvian13, MarvianNC, Marvian14, Marvian16, Piani16, Bu17}, and quantum thermodynamics \cite{Janzing00, Brandao13, Cwiklinski15, LostaglioNC, LostaglioX}, all of which have been studied within this framework.

Baumgratz et al. \cite{Baumgratz14} were the first to propose a general framework to quantify quantum coherence using the off-diagonal elements of density matrices  defined with respect to some preferred basis. This proposal was then followed up and further developed by subsequent research \cite{Winter16, Streltsov16, Chitambar16b, Tan16}.
It was noted that the preferred basis should be carefully chosen to define physically relevant incoherent states in the resource theory \cite{Chitambar16b, Streltsov16}.
In parallel, there have also been various approaches for quantifying quantum coherence in the context of asymmetry \cite{Gour08, Marvian14, Marvian16}.
The resource theory of asymmetry was first proposed to quantify the degree of symmetry breaking of a state under a group transformation
\cite{MarvianNC, Marvian14, Marvian16}.
It has applications in reference frame alignment \cite{Skotiniotis12}, quantum metrology \cite{Giovannetti04}, and quantum speed limits \cite{MarvianQSL}.                                                                                                           
A more rigorous investigation of asymmetry can be performed by dividing its contributions into independent modes. This approach is the so-called ``modes of asymmetry,''  which allows the concept to be applied to arbitrary finite and compact Lie groups \cite{Marvian14}.
In the case of specific group translations generated by an observable with non-degenerate eigenvalues, the 
coherence defined in Ref.~\cite{Baumgratz14} coincides with the translational asymmetry assuming that the eigenbasis of the observable is taken as the preferred basis.
The relationship between these different approaches to coherence \cite{Marvian16} can also be understood by distinguishing them into {\it speakable} and {\it unspeakable} notions.

Independently, there have been a number of studies that considered the notion of non-classicality  with macroscopic superpositions. A quantum superposition in microscopic scales can be observable today, but the question of the superposition principle in the macroscopic world remains, as  illustrated by Schr\"odinger's famous cat \cite{Schrodinger35}.
In this regard, various experiments have been performed to generate and observe such large-size quantum superpositions \cite{Brune96, Monroe96, Arndt99, Ourjoumtsev07, Schoelkopf14, Kasevich15}.
To quantify the size of a macroscopic superposition, or the so-called  `quantum macroscopicity', 
a number of proposals have been made \cite{Lee11, Frowis12, Nimmrichter13, Jeong15, Park16, Yadin16, Kwon16, Frowis17}, e.g., based on phase space structures \cite{Lee11, Park16}, quantum Fisher information \cite{Frowis12}, and the minimal extension of quantum theory \cite{Nimmrichter13}. 
Recently, an axiomatic approach to quantum macroscopicity that uses a similar framework to that of coherence~\cite{Baumgratz14} was also suggested \cite{Yadin16}, but it was shown to be not sufficient for certain types of states \cite{Kwon16}, although they may still be necessary.

In this paper, we provide an understanding of quantum macroscopicity through the lens of coherence and asymmetry. 
As a consequence, our study provides an avenue for the interpretation of quantum macroscopicity as a simple weighted sum of coherence and asymmetry.
We first suggest a coherence measure that satisfies the conditions in Ref.~\cite{Baumgratz14}.
It quantifies the asymmetry of a state when the eigenbasis of a group generator is taken as the preferred basis.
Based on the intuition from the phase-space-based measure of quantum macroscopicity \cite{Lee11, Park16}, we then introduce a general form of macroscopic coherence by taking sums of the form (effective size) $\times$ (degree of coherence) for all possible modes. 
Interestingly, by taking a specific form for the effective size of modes based on eigenvalue spacing of a given observable, our measure becomes an asymmetry measure with respect to a group transformation given by the observable.

Furthermore, we point out that a collective observable in a macroscopic system, otherwise called a {\it macroscopic observable} \cite{Poulin05} gives rise to the measure of quantum macroscopicity satisfying all the conditions in Ref.~\cite{Yadin16}.
We also discuss how microscopic coherences of a product state can be distinguished from macroscopic superposition in order to overcome inconsistencies of the framework suggested in Ref.~\cite{Yadin16}.
To do this, we introduce a scale parameter that reflects the inherent fuzziness of a measurement. 
By adjusting the parameter, we demonstrate that microscopic coherences from all product states can be excluded and only the contribution from macroscopic superposition remains.
Finally, we apply our measures in $N$-particle spin systems to demonstrate the validity of our approach, and we also study how decoherence can negatively impact quantum macroscopicity.
We emphasize that 
our approach provides an intuitive expression of quantum macroscopicity by capturing both the ``degree of coherence (quantum)" and the ``size of the system (macroscopicity)" in a single measure that simulaneously fits within the framework of quantum resource theories.

\section{Quantifying Quantum coherence}
\subsection{Quantifying coherence}
The resource theory of coherence can be constructed by a set of incoherent states ${\cal I}$ and incoherent operations ${\cal E}_{\rm IC}$ \cite{Baumgratz14}.
For a given basis $\{ \ket{i} \}$, an incoherent state $\hat{\delta} \in {\cal I} $ contains only diagonal terms, i.e. $\hat{\delta} = \sum_i p_i \ket{i} \bra{i}$, where $p_i \geq 0$ and $\sum_i p_i =1$.
An incoherent operation can be characterized by  Kraus operators as ${\cal E}_{\rm IC}(\hat{{\rho}}) = \sum_n \hat{K}_n \hat{{\rho}} \hat{K}^\dagger_n$, where $\hat{K}_n {\cal I} \hat{K}_n^\dagger \subseteq {\cal I}$ and $\sum_n \hat{K}_n^\dagger \hat{K}_n = {\mathbb 1}$.
The conditions that should be satisfied by a coherence measure ${\cal C}({\hat{{{\rho}}}})$ based on Ref.~\cite{Baumgratz14} are as follows.
(C1) ${\cal C}( \hat{\delta}) \geq 0$ and ${\cal C}( \hat{\delta}) = 0$ if and only if $ \hat{\delta} \in {\cal I} $.
(C2) Monotonicity under both
(C2a) a trace-preserving incoherent operation  ${\cal C}({\hat{\rho}}) \geq {\cal C}({\cal E}_{\rm IC} ({\hat{\rho}}))$ and
(C2b) a selective operation in average  ${\cal C}({\hat{\rho}}) \geq \sum_n p_n {\cal C}(\hat{K}_n {\hat{\rho}} \hat{K}^\dagger_n / p_n)$, where $p_n = {\rm Tr}\hat{K}_n^\dagger {\hat{\rho}} \hat{K} $.
(C3) Convexity  $\sum_i p_i {\cal C}({\hat{\rho}}_i) \geq {\cal C}(\sum_i p_i {\hat{\rho}}_i)$.

The overall sum of coherence between all the basis states, say $\ket{i}$ and $\ket{j}$ with $i\neq j$, quantifies quantum coherence of the state under consideration.
In this manner, the $l_1$-norm ${\cal C}_{l_1}({\hat{\rho}}) = \sum_{i \neq j} |{{\rho}}_{ij}|$, where ${{\rho}}_{ij} = \bra{i} {\hat{\rho}} \ket{j}$ was suggested as a proper measure of coherence and shown to satisfy (C1)--(C3) \cite{Baumgratz14}.
The $l_2$-norm given by ${\cal C}_{l_2}({\hat{\rho}}) = \sum_{i \neq j} |{{\rho}}_{ij}|^2$, however, does not satisfy (C2b) \cite{Baumgratz14}.
A geometric measure, which quantifies how far the state is apart from incoherent states, can also be a coherence measure.
For instance, the quantum relative entropy ${\cal C}_R({\hat{\rho}}) = \min_{\hat{\delta} \in {\cal I}} S({\hat{\rho}} || \hat{\delta}) = S({\hat{\rho}} || {\hat{\rho}}_{\rm diag})$ is a relevant measure of coherence, where ${\hat{\rho}}_{\rm diag} = \sum_i {{\rho}}_{ii} \ket{i}\bra{i}$ and $S({\hat{\rho}}||\hat\tau) = \Tr ({\hat{\rho}} \ln {\hat{\rho}} - {\hat{\rho}} \ln \hat\tau)$.
However, some geometric measures of coherence based on the Bures distance and the Hilbert-Schmidt norm have been found not to satisfy condition (C2b) \cite{Shao15}.
Recently, noticing the connection between entanglement and coherence \cite{Streltsov15}, a geometric measure of the type ${\cal C}_F({\hat{\rho}}) := 1-\max_{\hat{\delta} \in {\cal I}} {F}({\hat{\rho}}, \hat{\delta})$ has been proven to satisfy (C1)-(C3), where ${F}({\hat{\rho}}, \hat\tau) = [\Tr \sqrt{\sqrt{{\hat{\rho}}} \hat\tau \sqrt{{\hat{\rho}}}}]^2$ is the fidelity between quantum states.

\subsection{Coherence measure based on quantum affinity}
The similarity of information-theoretical properties between fidelity ${F}({\hat{\rho}},\hat\tau)$ and quantum affinity $A({\hat{\rho}},\hat\tau)^2 = [\Tr \sqrt{{\hat{\rho}}}\sqrt{\hat\tau}]^2$ was studied in Ref.~\cite{Luo04}.
Based on this observation, we introduce the following coherence measure,
\begin{equation}
{\cal C}_a({\hat{\rho}}) = 1- \max_{\hat{\delta} \in {\cal I}} A({\hat{\rho}},\hat{\delta})^2.
\end{equation}
Equivalent expressions of this measure are
\begin{equation}
\label{CMdef}
{\cal C}_a({\hat{\rho}}) = \sum_{i \neq j} |(\sqrt{\hat{\rho}})_{ij}|^2 = 1 - \sum_i (\sqrt{\hat{\rho}})_{ii}^2.
\end{equation}
This can be shown as follows. From the definition of the incoherent state, $\hat{\delta} = \sum_i p_i \ket{i} \bra{i}$, we have
$A({\hat{\rho}}, \hat{\delta})
= \sum_{i} \sqrt{p_i} (\sqrt{{\hat{\rho}}})_{ii} \leq \sqrt{\sum_i p_i} \sqrt{\sum_i (\sqrt{{\hat{\rho}}})_{ii}^2} = \sqrt{\sum_i (\sqrt{{\hat{\rho}}})_{ii}^2}$
by Cauchy-Schwartz inequality. The equality is achieved when $p_i = (\sqrt{\hat{\rho}})_{ii}^2 / \sum_i (\sqrt{\hat{\rho}})_{ii}^2$.
Hence we get $\max_{\hat{\delta} \in {\cal I}} A({\hat{\rho}}, \hat{\delta})^2 = \sum_i (\sqrt{\hat{\rho}})_{ii}^2 = 1 - {\cal C}_a({\hat{\rho}})$, which completes the proof.

These equivalent expressions show that the measure ${\cal C}_a$ captures the properties of both an interference-based and a geometric-based measure of coherence.
An interesting remark is that even though the $l_2$-norm for $\hat{\rho}$ does not satisfy condition (C2b),
the $l_2$-norm for $\sqrt{\hat{\rho}}$ obeys the condition.
Moreover, the measure ${\cal C}_a$ is bounded between $0$ and $1$.
The measure can also be efficiently computable since the coherence of a state in any given basis can be obtained by computing only the diagonal terms $(\sqrt{{\hat{\rho}}})_{ii}$.
Finally, the following theorem verifies that the measure ${\cal C}_a$ is a proper measure of quantum coherence.

\begin{theorem} [Affinity-based measure of coherence] With respect to a basis set $\{ \ket{i} \}$, ${\cal C}_a({\hat{\rho}}) = 1- \max_{\hat{\delta} \in {\cal I}}A^2({\hat{\rho}},\hat{\delta}) = \sum_{i \neq j} |(\sqrt{\hat{\rho}})_{ij}|^2$
is a measure of quantum coherence satisfying conditions (C1) -- (C3).
\label{Ca}
\end{theorem}
\begin{proof}
(C1)  can be easily checked that  ${\cal C}_a({\hat{\rho}}) = 0$ iff ${\hat{\rho}}$ only contains diagonal terms.
(C2a) can be proven by using the property $A({\hat{\rho}},\hat\tau) \leq A( {\cal E} ({\hat{\rho}}),{\cal E} (\hat\tau))$ for any trace-preserving map ${\cal E}$.
We then have ${\cal C}_a({\hat{\rho}}) = 1 - A({\hat{\rho}}, \hat{\delta}^*)^2 \geq 1 - A({\cal E}_{\rm IC}({\hat{\rho}}), {\cal E}_{\rm IC}(\hat{\delta}^*))^2 \geq 1 -  \max_{\hat{\delta} \in {\cal I}}A({\cal E}_{\rm IC} ({\hat{\rho}}), \hat{\delta})^2 = {\cal C}_a({\cal E}_{\rm IC}({\hat{\rho}}))$, where $\hat{\delta}^*$ maximizes $A({\hat{\rho}}, \hat{\delta})$ for $\hat{\delta} \in {\cal I}$.
(C2b) can be proven by showing $\sum_n p_n {\cal C}_a(\hat{K}_n {\hat{\rho}} \hat{K}^\dagger_n /p_n) \leq {\cal C}_a({\hat{\rho}})$, for an incoherent operator set $\{\hat{K}_n\}$.
 We first show that
$A({\hat{\rho}}, \hat\tau) \leq \sum_n A(\hat{K}_n {\hat{\rho}} \hat{K}^\dagger_n, \hat{K}_n \hat\tau \hat{K}^\dagger_n)$ for Kraus operators with $\sum_n \hat{K}^\dagger_n \hat{K}_n = {\mathbb 1}$ (see Appendix \ref{AppA}).
We then have 
\begin{equation}
\begin{aligned}
\sum_n p_n {\cal C}_a(\hat{K}_n {\hat{\rho}} \hat{K}^\dagger_n / p_n)
&= 1 - \sum_n p_n \max_{\hat{\delta}_n \in {\cal I}} A(\hat{K}_n {\hat{\rho}} \hat{K}^\dagger_n/p_n, \hat{\delta}_n)^2 \\
&\leq  1 - \sum_n A(\hat{K}_n {\hat{\rho}} \hat{K}^\dagger_n, \hat{K}_n \hat{\delta}^* \hat{K}^\dagger_n / q_n)^2 \\
&=  1 - \sum_n \frac{1}{q_n} A(\hat{K}_n {\hat{\rho}} \hat{K}^\dagger_n, \hat{K}_n \hat{\delta}^* \hat{K}^\dagger_n)^2 \\
&\leq  1 - \left[ \sum_n A(\hat{K}_n {\hat{\rho}} \hat{K}^\dagger_n, \hat{K}_n \hat{\delta}^* \hat{K}^\dagger_n) \right]^2 \\
&\leq  1 - A({\hat{\rho}},\hat{\delta}^*)^2 \\
&={\cal C}_a({\hat{\rho}}),
\end{aligned}
\end{equation}
where $q_n = \Tr \hat{K}_n \hat{\delta}^* \hat{K}^\dagger_n$, and $\hat\delta^*$ gives the maximum value of $A({\hat{\rho}},\hat{\delta})$ for $\hat{\delta} \in {\cal I}$.
Finally, (C3) can be proven by noticing that ${\cal C}_a({\hat{\rho}}) = 1 -  \sum_i  (\sqrt{\hat{\rho}})_{ii}^2 = 1 - \sum_i \Tr \sqrt{\hat{\rho}} \hat{P}_i \sqrt{\hat{\rho}} \hat{P}_i$, where $\hat{P}_i = \ket{i} \bra{i}$.
According to Lieb's concavity theorem \cite{Lieb73}, $\Tr \sqrt{\hat{\rho}} \hat{P}_i \sqrt{\hat{\rho}} \hat{P}_i$ is  then concave in ${\hat{\rho}}$ for all $i$, which makes ${\cal C}_a({\hat{\rho}})$ convex.
\end{proof}

\section{Quantum coherence and asymmetry}
\subsection{Quantifying asymmetry}

A resource theory of asymmetry can be constructed via translationally covariant operations \cite{Gour08, MarvianNC}.
In quantum mechanics, $U(1)$ group translations can be generated by a given observable $\hat{L}$ via
\begin{equation}
{\cal U}_x({\hat{\rho}}) = \hat{U}_x {\hat{\rho}} \hat{U}_x^\dagger,
\end{equation}
where $\hat{U}_x = e^{- i x \hat{L}}$ and $x  \in {\rm I\!R}$.
An observable $\hat{L}$ can be expressed using the eigen-decomposition $\hat{L} = \sum_i \lambda_i \ket{i} \bra{i}$, where $\lambda_i$ assigns some physical quantity to the specific eigenstate $\ket{i}$.
For instance, if $\hat{L}$ is the Hamiltonian, $\lambda_i$ refers to an energy eigenvalue of the system, and the related group operation is a time translation.
With respect to the group translations $\hat{U}_x$, we can specify free states and free operations from a resource theoretic viewpoint.
The free states are in this case the translationally covariant states satisfying ${\cal U}_x (\hat{\rho}) = \hat\rho$,
and the free operations ${\cal E}$ are the translationally covariant operations satisfying
${\cal U}_x \circ {\cal E} = {\cal E} \circ {\cal U}_x $ for all $x$ \cite{MarvianNC, Marvian16}.

The degree of asymmetry can be quantified by some measure $\cal A(\hat\rho)$ which gives ${\cal A}(\hat\rho) = 0$ for translationally-covariant states and also monotonically decreases under the translationally-covariant operations \cite{Gour08}.
Examples of such measures of asymmetry include information based measures \cite{MarvianNC, Marvian16, Piani16, Bu17}, robustness of asymmetry \cite{Piani16}, and asymmetry-weight \cite{Bu17}.

In this paper, we introduce an interference-based measure of asymmetry by adapting the techniques discussed in the previous section.
\begin{theorem} [Interference-based measure of asymmetry] For a given observable $\hat{L} = \sum_i \lambda_i \ket{i} \bra{i}$,
${\cal A}_a({\hat{\rho}}) = \sum_{\lambda_i \neq \lambda_j} |(\sqrt{\hat{\rho}})_{ij}|^2$
is a measure of asymmetry with respect to $\hat{L}$.
\label{Aa}
\end{theorem}
The proof can be found in Appendix \ref{AppB}.
This measure depends on the amplitude of interference between the different eigenstates of $\hat{L}$, and it also includes the contribution by the degenerate eigenvalues of $\lambda_i$.

\subsection{Modes of asymmetry}

The asymmetry of quantum states has been studied more finely by decomposing a quantum state into modes defined with respect to the eigenvalue spacing  of the observable $\hat{L}$  \cite{Marvian14}.
Analogous to the Fourier decompositions in the classical signal processing, we can decompose quantum states and channels into different modes.
Through this, each particular mode can have its own structure, thus allowing asymmetry to be quantified in a more fine-grained manner \cite{Marvian14}.

Every quantum state $\hat\rho$ can be expressed as the following sum over modes defined via the distance $\omega$
\begin{equation}
\hat\rho = \sum_{\omega \in \Omega} {\hat{\rho}}^{(\omega)},
\end{equation}
where
${\hat{\rho}}^{(\omega)} = \sum_{\lambda_i - \lambda_j = \omega} {{\rho}}_{ij} \ket{i}\bra{j}$
and $\Omega$ is a set composed of every possible spacing between the eigenvalues (i.e. $\omega = \lambda_i - \lambda_j$) of the observable $\hat{L}$.
Using this mode decomposition, an equivalent expression of the aforementioned free states is given by ${\hat{\rho}} = {\hat{\rho}}^{(0)}$, for which interference between the different eigenvalue spacings does not exist \cite{Marvian14, Marvian16, Yadin16}.
The following is an alternative definition of free operations.
\begin{proposition}[Covariant operations for modes of asymmetry]
\label{ECov}
A quantum operation ${\cal E}$ is a translationally covariant operation if and only if ${\cal E}$ satisfies ${\cal E}(\hat{\rho}^{(\omega)}) = {\cal E}(\hat{\rho})^{(\omega)}$ for every mode $\omega$.
\end{proposition}
\begin{proof}
We observe that 
${\cal U}_x ({\hat{\rho}}) = e^{-i\hat{L}x} {\hat{\rho}} e^{i\hat{L}x} = \sum_{\omega \in \Omega} e^{-i \omega x} {\hat{\rho}}^{(\omega)}$ by taking eigenbases of the observable $\hat{L}$. Then we have $({\cal E} \circ {\cal U}_x) ({\hat{\rho}}) = {\cal E} \left(\sum_{\omega \in \Omega} e^{-i \omega x} {\hat{\rho}}^{(\omega)} \right) = \sum_{\omega \in \Omega} e^{-i\omega x}   {\cal E} \left( {\hat{\rho}}^{(\omega)} \right) $. On the other hand, we have $( {\cal U}_x \circ {\cal E}) ({\hat{\rho}}) = \sum_{\omega \in \Omega} e^{-i\omega x}   {\cal E} ({\hat{\rho}})^{(\omega)}$. Two expressions are equal for translations ${\cal U}_x$ for all $x$ if and only if ${\cal E}({\hat{\rho}}^{(\omega)}) = {\cal E}({\hat{\rho}})^{(\omega)}$ which completes the proof.
\end{proof}

Using modes of asymmetry, Marvian and Spekkens \cite{Marvian14, Marvian16} proposed a measure to quantify the degree of interference stored within mode $\omega$.
The measure is given by 
\begin{equation}
{\cal A}^{(\omega)}_{\rm tr} (\hat\rho) = \Vert {\hat\rho}^{(\omega)}\Vert_1,
\end{equation}
where $\Vert \cdot \Vert_1$ is the trace norm.
$\Vert {\hat\rho}^{(\omega)} \Vert_1$ is non-increasing under covariant operations for every $\omega$ \cite{Marvian14}.
Furthermore, it can be shown that any linear function of the modes forms a measure of asymmetry, i.e.
$
\label{CC}
\left\lVert \sum_{\omega \in \Omega} c(\omega) {\hat\rho}^{(\omega)} \right\rVert_1
$
is a measure of asymmetry for any complex function $c(\omega)$ \cite{Marvian16}.

Based on this result, together with the previously suggested interference-based measure, we introduce a different kind of mode decomposition given by the following:
\begin{equation}
{\cal A}_{\rm HS}^{(\omega)}({\hat{\rho}}) = \sum_{\lambda_i - \lambda_j = \omega} |(\sqrt{\hat{\rho}})_{ij}|^2 = (\Vert \sqrt{\hat\rho}^{(\omega)} \Vert_{\rm HS})^2,
\end{equation}
where $\Vert \cdot \Vert_{\rm HS}$ is the Hilbert-Schmidt norm.
Unlike the modes of asymmetry based on the trace norm  ${\cal A}^{(\omega)}_{\rm tr}(\hat\rho)$, however, some modes of ${\cal A}_{\rm HS}^{(\omega)}({\hat{\rho}})$ can increase by covariant operations (see Appendix \ref{AppC} for an example).

For  both measures ${\cal A}_{\rm tr}^{(\omega)}$ and ${\cal A}_{\rm HS}^{(\omega)}$, we observe that the total degree of asymmetry is given by the sum over $\omega$ with $\omega \neq 0$, i.e.
\begin{equation}
{\cal A}({\hat{\rho}})_{{\rm tr} (a)} = \sum_{\omega \in \Omega - \{0\} }  {\cal A}^{(\omega)}_{{\rm tr (HS)}} ({\hat{\rho}}),
\end{equation}
which is non-increasing by covariant operations.
In this case, we define each mode of coherence $A^{(\omega)}(\hat\rho)$ as $\omega$-coherence.

\section{Quantifying Quantum Macroscopicity}
\subsection{Quantum macroscopicity under covariant operations}

A macroscopic physical system involves a large number of particles or modes.
In order to quantify quantumness in a macroscopic system, it is natural to consider an observable, often called a {\it macroscopic observable} \cite{Poulin05},
representing some collective physical quantity of a composite system, such as a total Hamiltonian, momentum, angular momentum (or spin), and the center-of-mass position.
The choice of an appropriate observable depends on the character of the system and the physics in which we are interested.

We note that many {\it macroscopic observables} are generators of the (collective) group transformations in the 
macroscopic system.
For $N$-partite systems, generators of this type of group transformations may be expressed as 
$$
\hat{L} = \sum_{n=1}^N \hat{L}^{(n)},
$$
where $\hat{L}^{(n)}$ is a generator for each local party.
For instance, the total Hamiltonian $\hat{H}_{\rm tot} = \sum_{n=1}^N \hat{H}^{(n)}$ gives rise to time translation $e^{-i \hat{H}_{\rm tot} t}$, total angular momentum $\vec{J}_{\rm tot}= \sum_{n=1}^N \vec{J}^{(n)}$ gives rise to rotation $e^{-i\theta \vec{n} \vec{J}_{\rm tot}}$ along the axis $\vec{n}$, and the center of mass position $\hat{x}_{\rm cm} = \sum_{n=1}^N \hat{x}^{(n)} / N$ or total momentum $\hat{p}_{\rm tot} = \sum_{n=1}^N \hat{p}^{(n)}$ translates a conjugate parameter according to $e^{-i p_0 \hat{x}_{\rm cm}}$ or $e^{-i x_0 \hat{p}_{\rm tot}}$, respectively.
Moreover, the eigenvalues of the collective generators $\hat{L}$ may be highly degenerate since they are given by the sum of eigenvalues of each local generator $\hat{L}^{(n)}$.

In this sense, it is natural to consider the asymmetry relative to some macroscopic observable, and its relationship to quantum macroscopicity.
An attempt to relate microscopic and macroscopic coherence phenomena via the resource theoretic framework was proposed by Yadin and Vedral \cite{Yadin16}, by disallowing quantum operations that allow macroscopic coherence and microscopic coherence to be inter-converted. This is achieved by considering the modes of asymmetry via $\omega$.
In fact, we note by Proposition~\ref{ECov} that  free operations in this framework of quantum macroscopicity are equivalent to  translationally-covariant operations with respect to the given macroscopic observable.
In particular, the quantum Fisher information and the Wigner-Yanase-Dyson skew information are measures of asymmetry that have been proven to satisfy the conditions 
to quantify quantum macroscopicity suggested in Ref.~\cite{Yadin16}.

\subsection{Weighted measures of asymmetry}

Following the method of quantifying macroscopic quantum superposition within phase space presented in~\cite{Lee11},
we consider the characterization of quantum macroscopicity by performing a sum of the form (effective size) $\times$ (degree of coherence) for every mode.
In this scenario, the effective size of the coherence is supplied by the eigenvalue spacing $\omega$ of an observable $\hat{L}$
and the degree of coherence is given by the mode coherence (or asymmetry) for each $\omega$.

As such, we introduce the following weighted sum of $\omega$ coherence as a measure for quantifying quantum macroscopicity:
\begin{equation}
{\cal M}({\hat{\rho}}) = \sum_{\omega \in \Omega} f(\omega) {\cal A}^{(\omega)}({\hat{\rho}})
\end{equation}
for a given function $f(\omega)$, which characterizes the effective size of each mode $\omega$.

In order for the measure to be consistent, we require that $f(\omega) = 0 $ when $\omega = 0$ in order to ensure that ${\cal M}(\hat\rho) = 0$ when $\hat\rho$ is a translationally-covariant (i.e. free) state with respect to $\hat{L}$.
For example, suppose we make a simple choice of $f(\omega) = \omega^2/2 = |\lambda_i - \lambda_j|^2 /2 $ for the $\omega$-coherence measure ${\cal A}^{(\omega)}_{\rm HS}$. In this case, the weighted sum 
then gives rise to the Wigner-Yanse-Dyson skew information: $I_W({\hat{\rho}}, \hat{L}) = -(1/2) \Tr[\sqrt{\hat{\rho}},\hat{L}]^2$ \cite{Wigner63},
which has been pointed out as a potential candidate for measuring quantum coherence \cite{Girolami14} and quantum macroscopicity \cite{Yadin16}. Our approach then gives the skew-information-based measure of quantum macroscopicity with the interpretation of a weighted sum of mode coherences.

We generalize this concept by proposing 
the possible classes of weight functions $f(\omega)$ in order to construct consistent measures of macroscopicity via a weighted sum of $\omega$-coherences.

\begin{theorem}[Weighted measure of asymmetry for the Hilbert-Schmidt norm]
\label{QM1}
Suppose $f(\omega) = \omega^2 \int_{x\in {\cal X}} dx [{\rm sinc}(\omega x/2)]^2 g(x)$ for $g(x) \geq 0$ and ${\cal X} \subset {\rm I\!R}$. Then
\begin{equation}
{\cal M}_{\rm HS} (\hat\rho) = \sum_{\omega \in \Omega^+} f(\omega) {\cal A}^{(\omega)}_{\rm HS} ({\hat{\rho}})
% = \sum_{\omega \in \Omega^+} \omega^2  h(\omega) {\cal A}^{(\omega)}_{\rm HS} ({\hat{\rho}})
\end{equation}
is a convex measure and is a monotone under covariant operations, where
$\Omega^+$ is the set of positive $\omega \in \Omega$.
\end{theorem}
We present the proof in Appendix \ref{AppD}.
The above construction generalizes the Wigner-Yanase-Dyson skew-information-based measure $I_W({\hat{\rho}}, \hat{L})$, which can be retrieved by choosing $g(x) = \delta(x) $, where $\delta(x)$ is the Dirac-delta function. In this case, $f(\omega) = \omega^2$.

A trace-norm-based quantum macroscopicity using the trace norm can also be obtained constructed as follows, again by considering the sum over modes of the form (effective size) $\times$ (degree of coherence):
\begin{theorem} [Weighted measure of asymmetry  for the trace norm]
\label{M1}
For $f(\omega) \geq 0$ for all $\omega$ and $f(0) = 0$,
\begin{equation}
{\cal M}_{\rm tr}({\hat{\rho}}) = \sum_{\omega \in \Omega^+} f(\omega) {\cal A}_{\rm tr}^{(\omega)}({\hat{\rho}})
\end{equation}
is a convex measure and monotone under covariant operations.
\end{theorem}

The proof is straightforward from the fact that $ {\cal A}_{\rm tr}^{(\omega)}({\hat{\rho}})$  is convex and monotone under covariant operations for every $\omega$, and $f(\omega)$ is a nonnegative function.
Similarly, we may take $f(\omega) = \omega^2$ to construct a measure of quantum macroscopicity based on the trace norm, ${\cal M}_{\rm tr}({\hat{\rho}}) = \sum_{\omega \in \Omega^+} \omega^2  ||\hat\rho^{(\omega)}||_1$.

\subsection{Conditions for macroscopic coherence measures and scaled measure of coherence}
In the previous section, we discussed weighted measures of $\omega$-coherence for some weight function $f(\omega)$.
To quantify the ``macroscopic'' coherence of quantum states, we are additionally required to impose an ordering between different eigenvalue spacings $\omega$. To this end, we may take the effective size $f(\omega)$ to be {\it monotonically increasing when $\omega$ increases}.

An important requirement for consistent quantum macroscopicity measures is that   products of many microscopic superpositions should be distinguished from genuine superpositions of macroscopically distinct states \cite{Leggett80}.
Examples of such accumulation of microscopic coherences are Bose-Einstein condensates and superconductivity.
In this sense, the conditions for quantum macroscopicity suggested in Ref.~\cite{Yadin16} may not be sufficient to define consistent measures because there exist measures satisfying them that can give rise to higher degrees of quantum macroscopicity for product states ${\hat{\rho}}^{\otimes N}$ than the Greenberger–-Horne-–Zeilinger(GHZ) type entangled states when the former is a simple accumulation of coherence between microscopic states, while the latter superposes macroscopically distinct states~\cite{Kwon16}.
This implies that any given weight function $f(\omega)$ needs to be checked against this condition in order to yield a consistent macroscopic coherence measure.

Here, we introduce a particular class of weight functions, parametrized by the scaling parameter $\sigma$, that will enable us to distinguish GHZ states from product states. We call this  a  {\it scaled} measure of quantum coherence based on the Hilbert-Schmidt norm,  which cuts off the microscopic contribution to coherence by introducing a scale $\sigma$.

\begin{definition}[Scaled measure of quantum coherence] For a given scale parameter $\sigma >0$, the scaled measure of quantum coherence is defined as
\begin{equation}
{\cal M}_\sigma({\hat{\rho}}) = \sum_{\omega \in \Omega}
\left[
 1 - e^{-\frac{\omega^2}{8\sigma^2}}
\right] {\cal A}^{(\omega)}_{\rm HS} ({\hat{\rho}}).
\end{equation}
\end{definition}

It can be shown that the scaled measure ${\cal M}_\sigma$ is non-increasing under translationally-covariant operations by applying Theorem~\ref{QM1} with $g(x) = x^2 (\sqrt\pi \tau)^{-1} e^{-x^2/\tau^2}$ and taking $\tau = (\sqrt{2} \sigma)^{-1}$ .
The parameter $\sigma$ determines an effective cutoff of the weight.
To see this, note that for $\omega \lesssim \sigma$, the weight $1 - \exp[{-\omega^2/(8\sigma^2)}]$ is relatively small compared to the case of $\omega \gtrsim \sigma$.
This cutoff may be used to exclude the contribution by microscopic coherence.
In the limit where there is no cutoff imposed, i.e.,  $\sigma \rightarrow 0$, the scaled measure of coherence becomes ${\cal M}_\sigma(\hat{\rho}) \rightarrow {\cal A}_{a} (\hat\rho)$, which is the standard ``unweighted" measure of asymmetry.

This measure can also be interpreted as the deviation of a quantum state for a fuzzy reference frame \cite{Jeong14, Frowis16}.
Note that $M_\sigma(\hat\rho) = \int D_H({\hat{\rho}},{\cal U}_x ({\hat{\rho}})) (\sqrt{\pi} \tau)^{-1} e^{-x^2/\tau^2} dx$, so the scaled measure of coherence has the interpretation of the average Hellinger distance $D_H(\hat\rho, \hat\tau) = (1/2) {\rm Tr} [\sqrt{\hat\rho} - \sqrt{\hat\tau}]^2$ generated by a group transformation ${\cal U}_x$ over the broadening by the Gaussian distribution, when the alignment of the reference frame is imperfect.

The scaled measure of coherence is also related to the measurement process with a finite precision \cite{Kofler07, Frowis16} onto the eigenbasis of the macroscopic observable. For the given observable $\hat{L} = \sum_i \lambda_i \ket{i}\bra{i}$, the Gaussian smoothing of the projections $\hat{P}_i = \ket{i}\bra{i}$ is given by $\hat{P}_i \rightarrow \hat{Q}_x^\sigma = \sum_i \sqrt{q_i^\sigma(x)} \hat{P}_i$, where $q^\sigma_i(x) = (\sqrt{2 \pi} \sigma)^{-1} e^{-(x-\lambda_i)^2/(2\sigma^2)}$ with the domain $x \in (-\infty,\infty)$.
In this case, the effect of the imperfect measurement process, $\Phi_\sigma(\hat\rho) = \int dx \hat{Q}^\sigma_x \hat\rho \hat{Q}^{\sigma\dagger}_x$, can be captured via the measurement-induced disturbance suggested in Refs.~\cite{Frowis16, Kwon16}, which gives the lower bound of the scaled measure of coherence,
\begin{equation}
\label{NEQ}
\frac{1}{2} D_B({\hat{\rho}},\Phi_\sigma({\hat{\rho}})) \leq {\cal M}_\sigma({\hat{\rho}}) \leq 1-e^{-\frac{I_W({\hat{\rho}},\hat{L})}{4\sigma^2}},
\end{equation}
where $D_B({\hat{\rho}},\hat\tau) = 2 - 2 \sqrt{F ({\hat{\rho}},\hat\tau)}$ is the Bures distance.
The proof can be found in Appendix \ref{AppE}.

It is important to note that the skew information $I_W$ is additive for a product state $\otimes_{n=1}^N \hat\rho_{n}$ with respect to a collective operator $\sum_{n=1}^N \hat{L}^{(n)}$, i.e., $I_W(\otimes_{n=1}^N \hat\rho_{n}, \sum_{n=1}^N \hat{L}^{(n)}) = \sum_{n=1}^N I_W(\hat\rho_{n}, \hat{L}^{(n)}) \leq N L^2_{\rm max} /4$, where $L_{\rm max}$ is the maximum 
among the eigenvalue differences of $\hat{L}^{(n)}$.
%among the operator norms $|| \hat{L}^{(n)}||_\infty$ and $I_W(\hat\rho_{n} , \hat{L}^{(n)}) \leq ||\hat{L}^{(n)}||_{\rm max}^2$.
Then the upper bound of (\ref{NEQ}) becomes 
\begin{equation}
\label{ProdBound}
{\cal M}_\sigma(\otimes_{n=1}^N \hat\rho_{n}) \leq 1 - \exp[- N L_{\rm max}^2 / (4 \sigma)^2) ].
\end{equation}
If we take the cutoff to be $\sigma = \sqrt{N \ln N}$, we have $\lim_{N \rightarrow \infty} {\cal M}_\sigma(\otimes_{n=1}^N \hat\rho_{n}) \rightarrow 0$, regardless of the local state $\hat\rho_{n}$ when $L_{\rm max}$ is bounded by a finite value.
Consequently, by the convexity of ${\cal M}_\sigma$,
{\it microscopic coherences contained in any separable multi-partite state are ruled out for the cutoff $\sigma = \sqrt{N \ln N}$ in the large particle limit of $N \gg 1$}.
%{\color{blue}
%This result demonstrates that multipartite entanglement is a necessary condition for quantum macroscopicity. It might be relevant to multipartite quantum correlation such as genuine multipartite nonlocality \cite{Lim10, Lee13} and GHZ theorem \cite{Ryu14}.}
The bound (\ref{NEQ}) might be also be useful for the direct detection of quantum macroscopicity in laboratories with finite precision measurements \cite{Frowis16}.

We also show that a general form of a scaling function can be chosen such that 
${\cal M}_\sigma(\hat\rho) = \sum_{\omega \in \Omega} f(\omega^2 / \sigma^2) {\cal A}_{\rm HS}^{(\omega)}(\hat\rho)$ is an asymmetry monotone 
for a concave function $f(x) \geq 0$ that is monotonically increasing with $x \geq 0$ and $f(0)=0$.
In this case, by taking a collective observable $\sum_{n=1}^N \hat{L}^{(n)}$ and the cutoff $\sigma = \sqrt{N \ln N}$, we can rule out microscopic coherences from every separable state $\hat\rho_{\rm sep}  = \sum_i p_i \otimes_{n=1}^N \hat\rho_{n}^i $ in $N$-partite systems because
\begin{equation}
\begin{aligned}
{\cal M}_\sigma(\hat\rho_{\rm sep}) &\leq \sum_i p_i {\cal M}_\sigma (\otimes_{n=1}^N \hat\rho_{n}^i) \\
&= \sum_i p_i f \left( \frac{\sum_{\omega \in \Omega}  \omega^2 {\cal A}_{\rm HS}^{(\omega)}(\otimes_{n=1}^N \hat\rho_{n}^i )}{N\ln N} \right)\\
& =\sum_i p_i  f \left( \frac{ 2 I_W(\otimes_{n=1}^N \hat\rho_{n}^i , \sum_{n=1}^N \hat{L}^{(n)}) }{N\ln N} \right) \\
&\leq f\left(\frac{L_{\rm max}^2}{2\ln N} \right)
\end{aligned}
\end{equation}
 becomes zero when $N \rightarrow \infty$ for a bounded $L_{\max}$.\\

\section{Application to $N$-partite spin-$1/2$ systems}
\subsection{Spin coherent states and GHZ-states}
In this section, we investigate the quantum macroscopicity of an $N$-partite spin-$1/2$ system with respect to 
the total spin observable along the $z$ axis,
$\hat{S}_z = \sum_{n=1}^N \hat{s}^{(n)}_z $, where $\hat{s}_z^{(n)} = {\mathbb 1}_1 \otimes {\mathbb 1}_2 \otimes \cdots \otimes {\mathbb 1}_{n-1} \otimes (\hat{\sigma}_z /2) \otimes {\mathbb 1}_{n+1} \otimes \cdots \otimes {\mathbb 1}_N$ is the local spin observable with the Pauli operator $\hat\sigma_z$.
Consequently, $\hat{S}_z$ has an eigenvalue spectrum $\{-N/2, -N/2+1, \cdots, N/2-1, N/2\}$ and the maximum difference between eigenvalues is $\omega_{\rm max} = N$.

In order to test the consistency of our measure,
we first compare a class of product states of the form: 
$$
\ket{\theta,\phi} = (\cos(\theta/2) \ket{0} + \sin(\theta/2) e^{i\phi} \ket{1})^{\otimes N},
$$
which are the so called spin-coherent states. We compare this
 with the generalized GHZ state:
$$
 \ket{\psi_{\rm GHZ}} = \cos(\theta/2) \ket{0}^{\otimes N} + \sin(\theta/2) e^{i\phi} \ket{1}^{\otimes N}
$$
 for $\theta \in [0, \pi]$ and $\phi \in [0, 2\pi]$.
Each mode of asymmetry for the spin coherent state may then be verified to be
\begin{widetext}
\small
$$
\begin{aligned}
{\cal A}^{(\omega)}_{\rm tr}(\ket{\theta, \phi}) &= \sum_{k=\omega}^N \sqrt{\binom{N}{k} \binom{N}{k-w}} \cos^{2N-2k+\omega}(\theta/2) \sin^{2k-\omega}(\theta/2) \approx \frac{1}{2} \exp \left[{-\frac{\omega^2}{2 N\sin^2 \theta}} \right] {\rm erfc}\left(\frac{\omega - 2N\sin^2(\theta/2)}{ \sqrt{2 N \sin^2 \theta} }\right),\\
{\cal A}^{(\omega)}_{\rm HS}(\ket{\theta, \phi}) &= \sum_{k=\omega}^N \binom{N}{k} \binom{N}{k-w} \cos^{2(2N-2k+\omega)}(\theta/2) \sin^{2(2k-\omega)}(\theta/2)\approx \frac{1}{2 \sqrt{\pi  N \sin^2{\theta}}} \exp\left[ - \frac{\omega^2}{N\sin^2\theta} \right] {\rm erfc} \left( \frac{\omega- 2N\sin^2(\theta/2)}{\sqrt{N  \sin^2 \theta}} \right),
\end{aligned}$$ 
\end{widetext}
\normalsize
respectively,
where ${\rm erfc}(x) = (2/\sqrt{\pi}) \int_x^\infty e^{-t^2} dt$ is the complementary error function and the approximations are given when $N \gg 1$ using the normal approximation of binomial distributions.

\begin{figure*}[t]
\includegraphics[width=0.90 \textwidth]{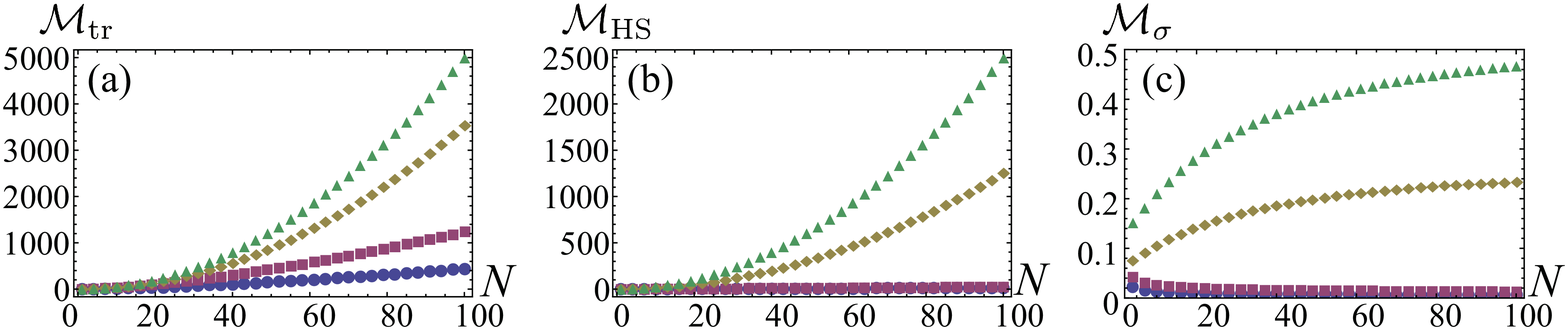} 
\caption{The degree of quantum macroscopcity for spin-coherent states $\ket{\theta,\phi}$ and GHZ-states $\ket{\psi_{\rm GHZ}}$ based on the weighted measures of (a) the trace norm (${\cal M}_{\rm tr}$), (b) the Hilbert-Schmidt norm (${\cal M}_{\rm HS}$),
 and (c) the scaled measure of coherence (${\cal M}_\sigma$) for $\sigma = \sqrt{N \ln N}$. 
Square symbols and circular symbols refer to the coherent states with $(\theta, \phi) = (\pi/2, 0)$, and $(\theta, \phi) = (\pi/4, 0)$, respectively.
Triangular symbols and diamond symbols refer to the GHZ-states with $\theta = \pi/2$, and $\theta=\pi/4$, respectively.}
\label{FIG1}
\end{figure*}

When we take the weight function to be $f(\omega) = \omega^2$ for the both measures,
we have 
 ${\cal M}_{\rm tr} (\ket{\theta, \phi}) \approx  \sqrt{ \pi / 2 } N^{3/2} \sin^3\theta \propto N^{3/2}$
and ${\cal M}_{\rm HS} (\ket{\theta, \phi}) = (1/4) N  \sin^2\theta \propto N$ for a large number of $N$.
On the other hand, for the GHZ state $\ket{\psi_{\rm GHZ}}$,
each mode of asymmetry is given by  
${\cal A}^{(\omega)}_{\rm tr}(\ket{\psi_{\rm GHZ}}) = (1/2) \sin \theta ( \delta_{N,\omega} + \delta_{N,-\omega} )$
and
${\cal A}^{(\omega)}_{\rm HS}(\ket{\psi_{\rm GHZ}}) = (1/4) \sin^2 \theta ( \delta_{N,\omega} + \delta_{N,-\omega} )$,
respectively.
In the case of $f(\omega)=\omega^2$, both weighted measures are given by 
${\cal M}_{\rm tr}(\ket{\psi_{\rm GHZ}}) = (1/2) N^2 \sin \theta$
and
${\cal M}_{\rm HS}(\ket{\psi_{\rm GHZ}}) = (1/4) N^2 \sin^2 \theta$,
which scale with $N^2$.
Figure \ref{FIG1} (a) and (b) show that 
the product state and the GHZ state scales differently with respect to the number of particles $N$, 
so microscopic coherence in the product state can be distinguished from macroscopic coherence in the GHZ-state in this manner.
Thus, for both measures, the choice of weight function $f(\omega) = \omega^2$ passes the basic consistency check, and they may be considered appropriate candidates for quantifying quantum macroscopcity, as GHZ states always have a larger macroscopicity than product states in the macroscopic limit.

We can also perform the check using the scaled measure of quantum coherence ${\cal M}_\sigma$ discussed previously.
The scaled measure of coherence for a spin-coherent state $\ket{\theta,\phi}$ is given by
${\cal M}_\sigma (\ket{\theta,\phi}) \approx 1 - \left[1+(N\sin^2{\theta}/(8 \sigma^2)) \right]^{-1/2}$ for $N \gg 1$.
Note that every spin-coherent state $\ket{\theta,\phi}$ is separable, thus the macroscopcity tends to ${\cal M}_\sigma (\ket{\theta,\phi}) \rightarrow 0 $ for a large value of $N \gg 1$ by Eq.~(\ref{ProdBound})
when $\sigma = \sqrt{N \ln N}$.
Figure~\ref{FIG1} (c) demonstrates how the product of microscopic coherence in a spin-coherent state $\ket{\theta,\phi}$ behaves differently from that of the GHZ-state, $\ket{\psi_{\rm GHZ}}$,
by taking the cutoff $\sigma = \sqrt{N \ln N}$.
On the other hand, the scaled measure of coherence for the GHZ-state is given by ${\cal M}_\sigma (\ket{\psi_{\rm GHZ}}) = (1/2) \sin^2\theta (1 - \exp[-N^2/(8\sigma^2)])$.
Thus, if we take $\sigma = \sqrt{N \ln N}$, the scaled measure for the GHZ state ${\cal M}_\sigma (\ket{\psi_{\rm GHZ}})$ gives a larger value for large $N$ (see Fig.~\ref{FIG1}), and so it also passes the consistency check.
This can be interpreted as evidence of genuine macroscopic coherence in the $N$-partite spin system.

\begin{figure*}[t]
\includegraphics[width=0.9 \textwidth]{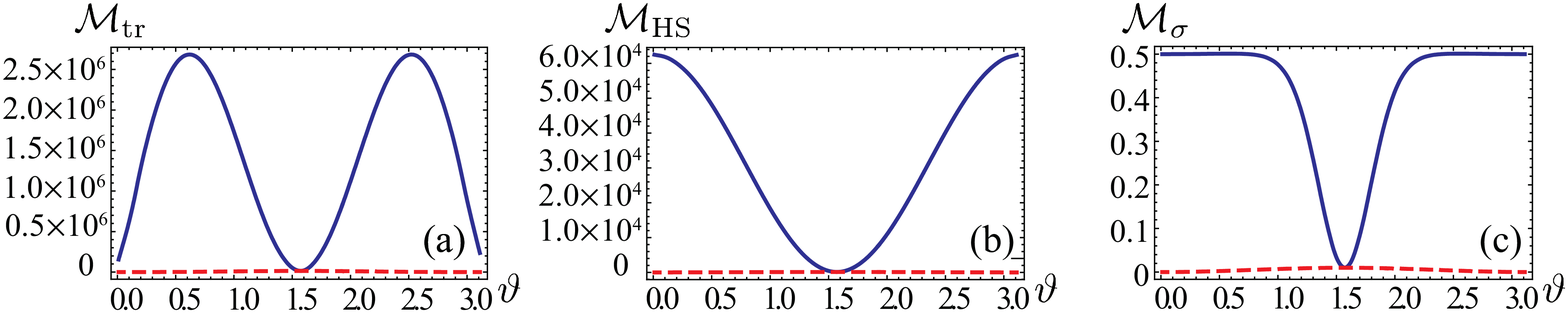} 
\caption{The degree of quantum macroscopcity for the product state $\ket{0}^{\otimes N}$ (dashed curves) and the GHZ-state $\ket{0}^{\otimes N} + \ket{1}^{\otimes N}$ (solid curves)
with respect to the total spin measurement $\hat{S}_{\vec{n}}$ for different $\vec{n} = (\sin{\vartheta}\sin{\varphi}, \sin\vartheta \cos\varphi, \cos{\vartheta})$. The weighted measures of (a) the trace norm (${\cal M}_{\rm tr}$) and (b) Hilbert-Schmidt norm (${\cal M}_{\rm HS}$)
 and (c) the scaled measure of coherence (${\cal M}_\sigma$) for $\sigma = \sqrt{N \ln N}$ are evaluated for $N=500$.
 All the measures do not depend on $\varphi$. }
\label{ThetaFig}
\end{figure*}
Quantum macroscopicity measures could be investigated 
for a general product state
$\ket{\psi_{\rm prod}} = \otimes_{i=1}^N \ket{\psi_i}$, where
$\ket{\psi_i} = \cos(\theta_i/2) \ket{0} + \sin(\theta_i/2) e^{i \phi_i} \ket{1}$
with general total spin measurement $\hat{S}_{\vec{n}} = \vec{n} \cdot \vec{{S}}$ with $\vec{n} = (\sin \vartheta \cos \varphi, \sin \vartheta\sin\varphi, \cos \vartheta)$ and $\vec{ S} := (\hat{S}_x, \hat{S}_y, \hat{S}_z)$.
In this case, the Hilbert-Schmidt norm based measure with the choice of weight function $f(\omega) = \omega^2$ is 
$
{\cal M}_{\rm HS}(\ket{\psi_{\rm prod}}) = \sum_{i=1}^N {\rm Var} (\ket{\psi_i}, \hat{s}_{\vec{n}}^{(i)}) = (1/4) \sum_{i=1}^N \sin^2{\Theta_i}, 
$
where $\Theta_i$ is an angle between two vectors $\vec{m}_i = (\sin\theta_i \cos\phi_i, \sin\theta_i \sin\phi_i, \cos\theta_i )$ and $\vec n$.
The scaled measure of coherence for
every product state, ${\cal M}_\sigma (\ket{\psi_{\rm prod}})$, vanishes by choosing the scale parameter $\sigma = \sqrt{N \ln N}$  when $N \gg 1$ by Eq.~(\ref{ProdBound}).

On the other hand, quantum macroscopicity of the GHZ state $\ket{\psi_{\rm GHZ}}$ may scale differently depending on the choice of the measurement basis.
The trace-norm-based measure for the GHZ state tends to oscillate by changing measurement axis $\vec{n}$, which gives the highest value for $\vartheta \approx \pi/4$.
The Hilbert-Schmidt norm-based measure of the GHZ state is given by
$
{\cal M}_{\rm HS}(\ket{\psi_{\rm GHZ}}) = (1/4) N^2 \sin^2{\theta} \cos^2{\vartheta} + (N/4) \sin^2{\vartheta}$
 for $N>2$.
Thus, quantum macroscopicity measures for the GHZ-state give significantly larger values than those for product states unless $\vartheta$ is near $\pi/2$.
When choosing the measurement axis $\hat{S}_x$ ($\vartheta = \pi/2,~\varphi=0$), however, quantum macroscopicity of the GHZ-state cannot be discriminated from product states since 
both states contain small degrees of coherence between distinct eigenstates of $\hat{S}_x$ with $\omega \propto N$, which vanishes when $N \gg 1$.

The scaled measure of coherence shows a similar behavior with the other two measures, but it seems more robust against the choice of the measurement basis $\vec{n}$.
In particular, in the limit of the large system size $N \gg 1$, the scaled measure of coherence for the GHZ-state is given by ${\cal M}_\sigma(\ket{\psi_{\rm GHZ}}) \approx 0.5$ except for the narrow region near $\vartheta = \pi/2$ while ${\cal M}_\sigma(\ket{\psi_{\rm prod}}) = 0$ for product states.
The difference of quantum macroscopicity between the product state $\ket{0}^{\otimes N}$ and the GHZ-state $\ket{0}^{\otimes N}+ \ket{1}^{\otimes N}$  is shown in Fig.~\ref{ThetaFig} with respect to the total spin measurement axis $\vec{n}$.

\subsection{Decoherence effect}

\begin{figure}[b]
\includegraphics[width= \linewidth]{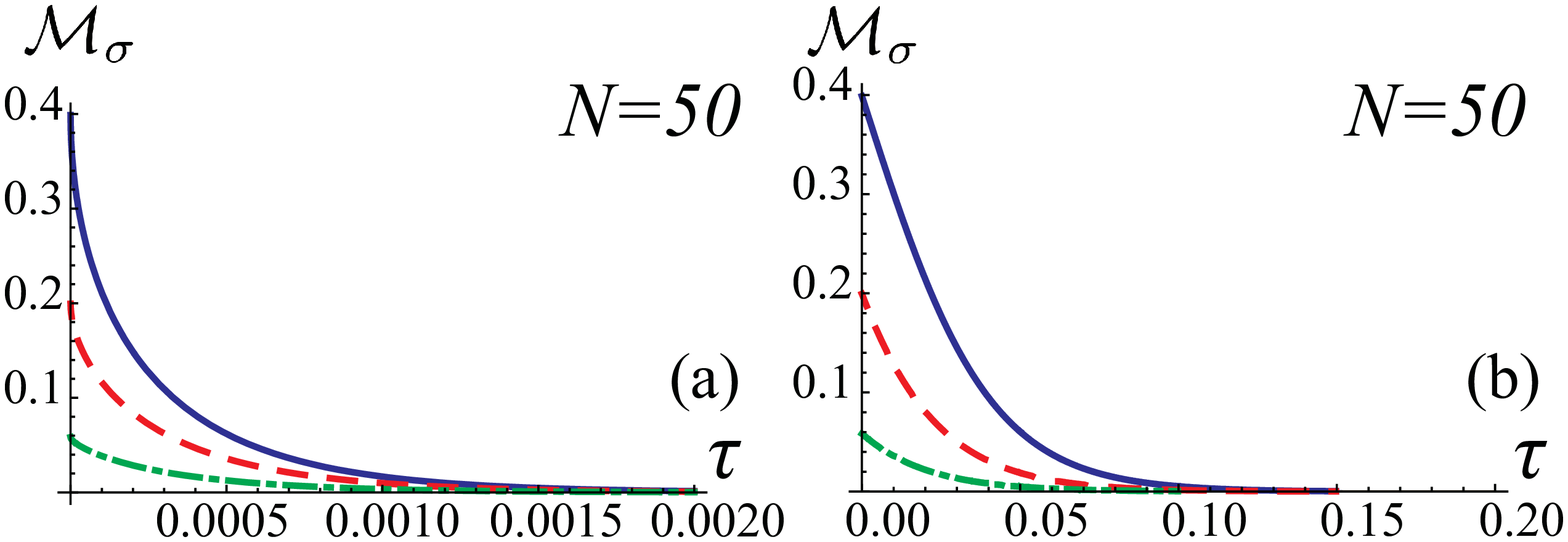} 
\caption{Decay of the scaled measure of quantum macroscopcity when the initial state is given by the pure GHZ-states $\ket{\psi_{\rm GHZ}} = \cos(\theta/2) \ket{0}^{\otimes N} + \sin(\theta/2)\ket{1}^{\otimes N}$.
The number of the particle is given by $N=50$.
Both 
(a) the dephasing channel $\hat{A}=\hat{S}_z$ and 
(b) the dissipation channel $\hat{A}=\hat{S}_-$ 
lead to the rapid decay of macroscopic coherence, even starting with the superposition $\ket{0}^{\otimes N} + \ket{1}^{\otimes N}$ $(\theta = \pi/2)$.
For both figures, solid curves refer to $\theta = \pi /2$, dashed curves  to $\theta = \pi / 4$, and dot-dashed curves to $\theta=\pi/8$.}
\label{FIG3}
\end{figure}

In this section, we study how decoherence affects the quantum macroscopicity for the $N$-particle spin system.
We analyze the degree of macroscopic coherence present in the system when
the system experiences a decoherence channel given by the following master equation in the Lindblad form,
\begin{equation}
\label{Lindblad}
 {\cal L}(\hat\rho) = \frac{d \hat\rho}{d \tau} =  \hat{A} \hat\rho \hat{A}^\dagger -\frac{1}{2} ( \hat{A}^\dagger \hat{A} \hat\rho + \hat\rho \hat{A}^\dagger \hat{A}).
\end{equation}

We first analyze the case of the dephasing channel, which corresponds to the case where $\hat{A} = \hat{S}_z$.
In this case, the GHZ-state $\ket{\psi_{\rm GHZ}}$ evolves according to
\small
$$
\begin{aligned}
\hat\rho_{\rm GHZ} (\tau) &=  \cos^2(\theta/2) (\ket{0} \bra{0})^{\otimes N} + \sin^2(\theta/2) (\ket{1} \bra{1})^{\otimes N} \\
&\quad+ \frac{\sin \theta}{2} e^{-\frac{N^2 \tau}{2}} ( e^{-i\phi} (\ket{0} \bra{1})^{\otimes N} + e^{i\phi} (\ket{1} \bra{0})^{\otimes N}),
\end{aligned}
$$
\normalsize
after time $\tau$.
Note that the off-diagonal terms experiences exponential decay $\exp[-N^2 \tau /2]$ so the quantum macroscopicity of the GHZ-state rapidly degrades under the dephasing channel.

We also study the dissipation channel described in the Lindblad form Eq.~(\ref{Lindblad}) by taking $\hat{A}= \hat{S}_-$,
where $\hat{S}_\pm = \sum_{n=1}^N \hat{s}_\pm^{(n)}$ are collective ladder operators given by the sum of ladder operators for each local party $\hat{s}_\pm^{(n)} = \hat{s}_x^{(n)} \pm i \hat{s}_y^{(n)}$.
Figure \ref{FIG3} demonstrates that quantum macroscopicity is fragile under both dephasing and dissipation channels.
These results imply that extremely noiseless environments are required 
in order to generate and manipulate quantum states while preserving macroscopic quantum coherence.
We also note that the degree of macroscopic coherence decays faster under the dephasing channel than the dissipation channel for a given characteristic time $\tau$ for the given parameters.

\section{Conclusion}
We introduced a measure of coherence that simultaneously quantifies asymmetry with respect to an observable
$\hat{L} = \sum_i \lambda_i \ket{i} \bra{i}$.
Our coherence and asymmetry measure can then be decomposed into modes given through the eigenvalue spacings $\omega = \lambda_i - \lambda_j$ specified by the observable $\hat{L}$. Using this, we may construct the so-called $\omega$-coherence and its corresponding measure. This allows us to discuss coherence,  asymmetry and macroscopic coherence on the same level.

We pointed out that quantum macroscopicity could be considered via the asymmetry with respect to some macroscopic observable that generates a collective group transformation on the total system.
As the system size gets larger, multiple modes $\omega$ contribute to the coherence and an ``effective size" of modes may be considered.
From this viewpoint, we defined a class of quantum macroscopicity measures from coherence and asymmetry measures characterized by a sum of the form (effective size) $\times$ (degree of coherence) for all the modes $\omega$. Through this, we demonstrated that many macroscopic measures of coherence may be related to the total coherence of a system via a simple weighted sum, of which the skew information based measure is a special case.
It will be interesting for future work to investigate whether previously studied quantum macroscopicity measures such as quantum Fisher information $I_F(\hat\rho,\hat{L}) = 2 \partial_x^2 D_B(\hat{\rho}, {\cal U}_x (\hat\rho))$ and generalized skew informations $I_\alpha(\hat{\rho}, \hat{L}) = (-1/2) {\rm Tr} [\hat\rho^{\alpha}, \hat{L}] [\hat\rho^{1-\alpha}, \hat{L}]$ can also be formulated in this framework.

We also discussed how it is desirable to exclude microscopic superpositions in order to implement a proper measure of macroscopic coherence, which is not guaranteed simply by the conditions proposed in Ref.~\cite{Yadin16} as shown in Ref.~\cite{Kwon16}. This necessarily imposes additional constraints on the weight function, which must be verified in order to yield a consistent measure.
We then introduced a scaled measure of coherence, where the coherence for each mode is differently weighted by a given scaling parameter $\sigma$.
This scaling parameter may be interpreted as a fuzziness in the reference frame, which rules out microscopic superpositions that are not detectable for a given degree of fuzziness.
In this way, the measure assures that only the coherence between macroscopically distinct states is considered.
We then compared the degree of quantum macroscopicity of a product state and a GHZ-state in $N$-particle spin systems.
We showed that the microscopic portion of the coherence present in product states is effectively suppressed by introducing the cutoff $\sigma = \sqrt{N\ln N}$. We also considered decoherence effects, and we demonstrated numerically that the degree of quantum macroscopicity present in the GHZ state is extremely susceptible to decoherence.

Our study develops the  conceptual notion of quantum macroscopicity
by accounting for both the ``degree of coherence (quantum)" and the ``size of the system (macroscopicity)" while simultaneously falling under the framework of the resource theory of asymmetry.
We stress that the arguments presented are not limited to any particular systems, but may also be applied to any general macroscopic observable $\hat{L}$ for any macroscopic, composite systems. We expect that our study would add insight into the general properties of genuine macroscopic quantum effects.

{\it Note added: }
Recently, we became aware of Ref.~\cite{Yu17}, in which the same type of coherence measure was suggested but without its extension to quantum macroscopicity measure.
We realized that Eq.~(\ref{CMdef}) in Theorem.~\ref{Ca} of our present manuscript is identical to Eq.~(2) of Ref.~\cite{Yu17}.

\section*{acknowledgment}
This work was supported by a National Research Foundation of Korea grant funded by the Korea government (Grant No. 2010-0018295) and by the Korea Institute of Science and Technology Institutional Program (Project No. 2E26680-16-P025).
K.C. Tan was supported by Korea Research Fellowship Program through the National Research Foundation of Korea (NRF) funded by the Ministry of Science and ICT (Grant No. 2016H1D3A1938100).

\appendix
\section{Complete proof of Theorem.~\ref{Ca}}
\label{AppA}
 We show that
$A({\hat{\rho}}, \hat\sigma) \leq \sum_n A(\hat{K}_n {\hat{\rho}} \hat{K}^\dagger_n, \hat{K}_n \hat\sigma \hat{K}^\dagger_n)$ for Kraus operator set $\sum_n \hat{K}^\dagger_n \hat{K}_n = {\mathbb 1}$.
\begin{proof}
A set of Kraus operators $\{ \hat{K} \}$ can be expressed using ancillary state $\hat\tau_2$: $\hat{K}_n {\hat{\rho}} \hat{K}^\dagger_n = \Tr_2 ({\mathbb 1} \otimes \hat{\Pi}_n) \hat{U} ({\hat{\rho}} \otimes \hat\tau_2) \hat{U}^\dagger  ({\mathbb 1} \otimes \hat{\Pi}_n)$.
Note that $A({\hat{\rho}}, \hat\sigma)$ is non-increasing under partial trace $A({\hat{\rho}}_{12}, \hat\sigma_{12}) \leq A(\Tr_2 {\hat{\rho}}_{12}, \Tr_2 \hat\sigma_{12})$ and satisfies the following properties for a set of projection operators $\{ \hat{\Pi}_n \}$: 
$\sum_n A(\hat{\Pi}_n {\hat{\rho}} \hat{\Pi}_n, \hat{\Pi}_n \hat\sigma \hat{\Pi}_n) = A(\sum_n \hat{\Pi}_n {\hat{\rho}} \hat{\Pi}_n, \sum_n \hat{\Pi}_n \hat\sigma \hat{\Pi}_n)$.
Using these properties, we can show that
\begin{widetext}
\begin{equation}
\begin{aligned}
\sum_n A(\hat{K}_n {\hat{\rho}} \hat{K}^\dagger_n, \hat{K}_n \hat\sigma \hat{K}^\dagger_n) 
&= \sum_n A(\Tr_2 ({\mathbb 1} \otimes \hat{\Pi}_n) \hat{U} ({\hat{\rho}} \otimes \hat\tau_2) \hat{U}^\dagger  ({\mathbb 1} \otimes \hat{\Pi}_n), \Tr_2 ({\mathbb 1} \otimes \hat{\Pi}_n) \hat{U} (\hat\sigma \otimes \hat\tau_2) \hat{U}^\dagger  ({\mathbb 1} \otimes \hat{\Pi}_n)) \\
&\geq \sum_n A(({\mathbb 1} \otimes \hat{\Pi}_n) \hat{U} ({\hat{\rho}} \otimes \hat\tau_2) \hat{U}^\dagger  ({\mathbb 1} \otimes \hat{\Pi}_n), ({\mathbb 1} \otimes \hat{\Pi}_n) \hat{U} (\hat\sigma \otimes \hat\tau_2) \hat{U}^\dagger  ({\mathbb 1} \otimes \hat{\Pi}_n)) \\
&=  A(\sum_n({\mathbb 1} \otimes \hat{\Pi}_n) \hat{U} ({\hat{\rho}} \otimes \hat\tau_2) \hat{U}^\dagger  ({\mathbb 1} \otimes \hat{\Pi}_n), \sum_n ({\mathbb 1} \otimes \hat{\Pi}_n) \hat{U} (\hat\sigma \otimes \hat\tau_2) \hat{U}^\dagger  ({\mathbb 1} \otimes \hat{\Pi}_n)) \\
&\geq  A(\hat{U} ({\hat{\rho}} \otimes \hat\tau_2) \hat{U}^\dagger, \hat{U} (\hat\sigma \otimes \hat\tau_2) \hat{U}^\dagger) \\
&=  A({\hat{\rho}} \otimes \hat\tau_2, \hat\sigma \otimes \hat\tau_2) \\&=  A({\hat{\rho}}, \hat\sigma) A(\hat\tau_2,\hat\tau_2) \\&=  A({\hat{\rho}}, \hat\sigma).
\end{aligned}
\end{equation}
\end{widetext}
\end{proof}
\section{Proof of Theorem.~\ref{Aa}}
\label{AppB}
For the non-degenerate case (i.e. $\lambda_i \neq \lambda_j$ if and only if $i \neq j$), the proof is the same as that for Theorem~\ref{Ca}.
In the case of degeneracy, we write a resource-free state $\hat\sigma = \sum_n p_n {\hat\sigma}_n = \sum_n p_n \sum_{\lambda} \lambda(n) \ket{n, \lambda} \bra{n, \lambda}$, where each $\hat\sigma_n$ is a translationally-covariant state and $\sum_{\lambda} \lambda(n) \ket{n, \lambda} \bra{n, \lambda}$ is its eigndecomposition.
Then we can follow the proof of Theorem~\ref{Ca} if we can always choose a set of bases $\{ \ket{n, \lambda} \}$ which gives ${\cal A}_a (\hat\rho) = 1 - \max_{\{p_n, \lambda(n), \ket{n,\lambda} \}} A^2(\hat\rho, \hat\sigma) = \sum_{\lambda_i \neq \lambda_j} |(\sqrt{\hat{\rho}})_{ij}|^2$.

Now we consider a projection $\hat{P}_n$ onto the states with $\lambda_i = n$. Using this projection, we can block-diagonalize $\sqrt{\hat\rho}$ and  take eigendecomposition of each block $\hat{P}_n (\sqrt{\hat\rho}) \hat{P}_n$ in order to obtain the desired free state.

\section{Increasing of ${\cal A}_{\rm HS}^{(\omega)}({\hat{\rho}})$ by a covariant operation}
\label{AppC}
We give an example of the case of increasing ${\cal A}_{\rm HS}^{(\omega)}({\hat{\rho}})$ by a covariant operation.
Consider the quantum state $\hat\rho = \ket{\psi}\bra{\psi}$, where $\ket{\psi} = 3^{-1/2} (\ket{0} + \ket{1} + \ket{2}) $,
\begin{equation}
\hat\rho = \frac{1}{3}
\left(
\begin{matrix}
1 & 1 & 1 \\
1 & 1 & 1 \\
1 & 1 & 1 
\end{matrix}
\right).
\end{equation}
Then we consider a partially-decohering map on $\omega = \pm1$, which is a translationally-covariant operation.
Under the operation, the state $\hat\rho$ evolves into
\begin{equation}
\Phi(\hat\rho) = \frac{1}{3}
\left(
\begin{matrix}
1 & 0 & 1 \\
0 & 1 & 0 \\
1 & 0 & 1 
\end{matrix}
\right).
\end{equation}
In this case, we can calculate each mode of coherences ${\cal A}_{\rm HS}^{(\omega)}$ for $\hat\rho$ and $\Phi(\hat\rho)$,
${\cal A}^{(\pm 1)}_{\rm HS} (\hat\rho) = 2/9$ and ${\cal A}^{(\pm 1)}_{\rm HS} (\Phi(\hat\rho)) = 0$,
while ${\cal A}^{(\pm 2)}_{\rm HS} (\hat\rho) = 1/9$ and ${\cal A}^{(\pm 2)}_{\rm HS}(\Phi(\hat\rho)) = 1/6$.
Thus, for $\omega = 2$ we note that the mode of coherence is increased by a translationally-covariant operation.

Meanwhile, the total asymmetry decreases under the partial decohering map: ${\cal A}_a(\rho) = \sum_{\omega \in \{ \pm 1 \pm 2\} } {\cal A}_{\rm HS}^{(\omega)} (\hat\rho) = 2/3$, while ${\cal A}_a(\Phi(\hat\rho)) = \sum_{\omega \in \{ \pm 1 \pm 2\} } {\cal A}_{\rm HS}^{(\omega)} (\Phi(\hat\rho)) =  1/3$.

\section{Proof of Theorem.~\ref{QM1}}
\label{AppD}
We first show that the following construction is possible using the modes of asymmetry.
\begin{proposition} [Weighted measure of asymmetry]
\begin{equation}
\sum_{\omega \in \Omega} \left( 1 - e^{-i\omega x} \right) {\cal A}^{(\omega)}_{\rm HS} ({\hat{\rho}})
\label{WeightedNC}
\end{equation}
is a convex measure and monotone under covariant operations for every $x$.
\end{proposition}
\begin{proof}
We note that the Hellinger distance between a quantum state and its symmetric transformation
\begin{equation}
D_H({\hat{\rho}},{\cal U}_x ({\hat{\rho}})) = 1 - \Tr[\sqrt{\hat{\rho}} e^{-i\hat{L}x} \sqrt{\hat{\rho}} e^{i\hat{L}x}]
\end{equation}
is a measure of asymmetry, i.e. convex and non-increasing under translationally-covariant operations \cite{MarvianNC} for any $x \in {\rm I\!R}$.
Then by direct expansion on the eigenbasis of $\hat{L}$, we get $D_H({\hat{\rho}},{\cal U}_x ({\hat{\rho}})) = 1 - \sum_{i,j} |(\sqrt{{\hat{\rho}}})_{ij}|^2 e^{-i(\lambda_i-\lambda_j)x} =  \sum_{\omega \in \Omega} (1- e^{-i \omega x}) {\cal A}^{(\omega)}_a({\hat{\rho}})$.
\end{proof}

Note that ${\cal A}^{(-\omega)}_{\rm HS} ({\hat{\rho}}) = {\cal A}^{(\omega)}_{\rm HS} ({\hat{\rho}})$ by the hermicity of the density matrix, so the above quantity will always give rise to real values.
To make this explicit, we may alternatively perform the sum over $\Omega^+$, which is the set of positive $\omega$ in $\Omega$. We then have $\sum_{\omega \in \Omega} (1- e^{-i \omega x}) {\cal A}^{(\omega)}_{\rm HS} ({\hat{\rho}}) = x^2  \sum_{\omega \in \Omega^+} \omega^2 [ {\rm sinc}(\omega x /2) ]^2 {\cal A}^{(\omega)}_{\rm HS}$.
Then we note that the integration on $x$ with multiplying a well-defined function $g(x)/x^2 \geq 0$,
$$
\begin{aligned}
&\int dx  g(x) \sum_{\omega \in \Omega^+} \omega^2 [ {\rm sinc}(\omega x /2) ]^2 {\cal A}^{(\omega)}_{\rm HS} \\
&\quad = \sum_{\omega \in \Omega^+} \omega^2  \int dx [ {\rm sinc}(\omega x /2) ]^2  g(x)  {\cal A}^{(\omega)}_{\rm HS} 
\end{aligned}
$$
does not change the monotonicity and convexity.
Finally by defining $f(\omega) =  \omega^2 \int dx [ {\rm sinc}(\omega x /2) ]^2  g(x)$, ${\cal M}_{\rm HS}(\hat\rho) = \sum_{\omega \in \Omega^+} f(\omega) {\cal A}_{\rm HS}^{(\omega)} (\hat\rho)$ becomes a convex measure, which is monotone under covariant operations.

\section{Proof of the bound Eq.~(\ref{NEQ})}
\label{AppE}
Using the relation between the fidelity and the affinity \cite{Luo04}, we note that 
\begin{equation}
\frac{1}{2} D_B({\hat{\rho}},\Phi_\sigma({\hat{\rho}})) \leq D_H({\hat{\rho}},\Phi_\sigma({\hat{\rho}})).
\end{equation}
The first inequality of Eq.~(\ref{NEQ}) can then be proved by
\begin{equation}
\begin{aligned}
D_H({\hat{\rho}},\Phi_\sigma({\hat{\rho}})) 
&= 1 - {\rm Tr} \sqrt{\hat\rho} \sqrt{  \int dx \hat{Q}^\sigma_x \hat\rho \hat{Q}^{\sigma\dagger}_x} \\
&\leq 1 - \int dx {\rm Tr} \sqrt{\hat\rho} \hat{Q}^\sigma_x \sqrt{\hat\rho} \hat{Q}^{\sigma\dagger}_x \\
&= 1 - \sum_{i,j} \int dx \sqrt{q^\sigma_i(x) q^\sigma_j(x)} {\rm Tr} \sqrt{\hat\rho} \hat{P}_i \sqrt{\hat\rho} \hat{P}_j \\
&= 1 - \sum_{i,j} e^{-\frac{(\lambda_i - \lambda_j)^2}{8\sigma^2}}  {\rm Tr} \sqrt{\hat\rho} \hat{P}_i \sqrt{\hat\rho} \hat{P}_j \\
&= \sum_{\omega \in \Omega} \left[ 1 -  e^{-\frac{\omega^2}{8\sigma^2}} \right] \sum_{\lambda_i - \lambda_j = \omega} {\rm Tr} \sqrt{\hat\rho} \hat{P}_i \sqrt{\hat\rho} \hat{P}_j \\
&={\cal M}_\sigma(\hat\rho),
\end{aligned}
\end{equation}
where the inequality comes from operator Jensen's inequality \cite{Hansen03} and noting that $\hat{Q}^\sigma_x = \hat{Q}^{\sigma\dagger}_x$ are unital operators. Also note that 
$\int dx \sqrt{q^\sigma_i(x) q^\sigma_j(x)} = \exp[-(\lambda_i - \lambda_j)^2/(8\sigma^2)]$ for $q_i^\sigma(x) = (\sqrt{2 \pi} \sigma)^{-1} e^{-(x-\lambda_i)^2/(2\sigma^2)}$ and $\sum_{\lambda_i - \lambda_j = \omega} {\rm Tr} \sqrt{\hat\rho} \hat{P}_i \sqrt{\hat\rho} \hat{P}_j = \sum_{\lambda_i - \lambda_j = \omega} |(\sqrt{\hat\rho})_{ij}|^2 = {\cal A}_{\rm HS}^{(\omega)}(\hat\rho)$.

The second inequality holds by Jensen's inequality
\begin{equation}
\begin{aligned}
{\cal M}_\sigma(\hat\rho)
&= \sum_{\omega \in \Omega} \left[  1 - e^{-\frac{\omega^2}{8\sigma^2}} \right] {\cal A}^{(\omega)}_{\rm HS} ({\hat{\rho}}) \\
&\leq 1 - e^{ -\sum_{\omega \in \Omega} \frac{\omega^2}{8\sigma^2} {\cal A}^{(\omega)}_{\rm HS} (\hat\rho)}  \\
&=  1-e^{-\frac{I_W({\hat{\rho}},\hat{L})}{4\sigma^2}},
\end{aligned}
\end{equation}
where $ 1 - e^{-\frac{\omega^2}{8\sigma^2}}$ is a concave function of $\omega^2$.

\end{document}